\documentclass[11pt,letterpaper]{article}
\usepackage{graphicx,amssymb,amsmath}

\usepackage{times}
\usepackage[T1]{fontenc}
\usepackage{ae,aecompl}
\usepackage{amssymb}
\usepackage{latexsym}
\usepackage{amsmath}
\usepackage{bbm}
\usepackage{psfrag}
\usepackage{color}
\usepackage[]{algorithm2e}
\usepackage{tikz}
\usetikzlibrary{calc,trees,positioning,arrows,chains,shapes.geometric,%
    decorations.pathreplacing,decorations.pathmorphing,shapes,%
    matrix,shapes.symbols}
\usepackage{mathrsfs}
\usepackage{amsthm}
\usepackage[margin=1in]{geometry}
\usepackage{authblk}
\usepackage{url}
\usepackage{paralist}
\usepackage{caption}

\newcommand{\conf}{\ensuremath{c}}

\newcommand{\per}{\ensuremath{t_{a, b}}} 
\newcommand{\sight}{\ensuremath{s}} 
\newcommand{\word}{\ensuremath{w}}

\newcommand{\Dhz}{\ensuremath{\Delta}_0}

\newcommand{\tot}{\ensuremath{m}}

\newcommand{\Spath}{\ensuremath{\mathbb{P}}}
\newcommand{\N}{\ensuremath{\mathbb{N}}}
\newcommand{\Z}{\ensuremath{\mathbb{Z}}}

\newcommand{\dyn}{\ensuremath{\delta}}

\newcommand{\DD}{\ensuremath{\Delta}}

\newcommand{\expect}{\ensuremath{\mathbb{E}}}

\newcommand{\lMin}{\ensuremath{h_{min}}}
\newcommand{\lMax}{\ensuremath{h_{max}}}
\newcommand{\DDl}{\ensuremath{\Delta_h}}

\newcommand{\lR}{\ensuremath{\ell_{r}}}
 
\newcommand{\lOpt}{\ensuremath{\ell_{opt}}}

\newtheorem{theorem}{Theorem}
\newtheorem{lemma}{Lemma}
\newtheorem{prop}{Proposition}
\newtheorem{definition}{Definition}
\newtheorem{property}{Property}
\newtheorem{corollary}{Corollary}

\title{Lost in Self-stabilization. \thanks{This work is partially supported by Programs ANR Dynamite, Quasicool and IXXI (Complex System Institute, Lyon).}}

\author[1]{Damien Regnault\thanks{Corresponding author.}}
\author[2]{\'Eric R\'emila}
\affil[1]{IBISC, EA4526, 

Universit{\'e} d'{\'E}vry Val-d'Essonne, 91037 
	{\'E}vry, France.
	
{\texttt{damien.regnault@ibisc.univ-evry.fr}}}

\affil[2]{Universit\'e  de Lyon, GATE LSE (UMR CNRS  5824),  

Site st\'ephanois, 42023 Saint-Etienne, France.

 \texttt{eric.remila@univ-st-etienne.fr}}

\begin{document}

\maketitle

\begin{abstract}

One of the questions addressed here is "How can a twisted thread correct itself?". We consider a theoretical model where the studied mathematical object represents a $2D$ twisted discrete thread linking two points. This thread is made of a chain of agents which are ``lost'', \emph{i.e.} they have no knowledge of the global setting and no sense of direction. Thus, the modifications made by the agents are local and all the decisions use only minimal information about the local neighborhood. We introduce a random process such that the thread reorganizes itself efficiently to become a discrete line between these two points. The second question addressed here is to reorder a word by local flips in order to scatter the letters to avoid long successions of the same letter. These two questions are equivalent. The work presented here is at the crossroad of many different domains such as modeling cooling process in crystallography \cite{FBR2010ap,FBR2010an,FBR2010ao}, stochastic cellular automata \cite{Fates11,FMST06}, organizing a line of robots in distributed algorithms (the robot chain problem \cite{DKLM06,KM09}), and Christoffel words in language theory \cite{Ber07}.
\end{abstract}

\newpage
\section{Introduction}

\subsection{The result}

\input{trace1.tex}

In this paper, we define and analyze a random process whose interest is at the crossroad of many different domains. Among all the different interpretations of our work, we choose a toy model using a simple graphical representation to ease the reading of the article. To avoid lengthy definitions, we start by a simulation of the random process that we study along this paper: see Figure \ref{trace1}. This simulation shows a twisted discrete thread reorganizing itself into a good approximation of the continuous line linking its two endpoints. In fact, this thread is made by a chain of agents. Movements of the agents correspond to local modifications (called flips) of the twisted thread. These movements will slowly but surely transform the twisted thread into a discrete line. The main interest of our result is that our process is heavily constrained. All modifications are decided locally by the agents which are memoryless, have no sense of direction, and no knowledge of the global setting. The decisions have to be made only with the relative position of the closest neighboring agents. We show that, despite all these constraints, it is possible to program the agents to achieve our goal.





Our model is precisely described in section 2, but we can now give an informal presentation. 
We work on a $2D$ discrete grid of size $A\times B$ (see Figure \ref{fig:grid}). We are given  chain of $A+B$ agents, which  forms a path  between the opposite endpoints  of the grid, of coordinates $(0, 0)$ and $(A, B)$. We want to design a rule such that, at each time step, an agent is randomly chosen and is allowed  to jump in an other site of the grid, with preservation of the connectivity of the chain. The goal of the process  is  to reorganize the chain such that it stabilizes in a position  as close as possible to the continuous line, of slope $\frac{B}{A}$, passing by the opposite  endpoints.

We design here a distributed algorithm which achieves it efficiently (in polynomial time according to the distance between the two endpoints) and uses only local information. If each agent has a local sight of $s$ (\emph{i.e.} it can only observe sites which are at distance at most  $s$ from its own position). 


Let $(t_a, t_b)$ be the pair of relatively prime positive integers such that $\frac{t_b}{t_a}= \frac{B}{A}$. 
We prove that our  process succeeds when  $t_a+t_b \leq s $. We also show that this bound on the sight is almost tight: a sight of at least $t_a+t_b-1$ is needed to self-stabilize a discrete line of slope $\frac{B}{A}$. 

We have another difficulty, due to boundary conditions. Our  process succeeds when the starting chain is completely on one  side of the grid limited by the line segment linking endpoints. For certain initial conditions,  we only reach a set of configurations, all very close to the target line,  but we cannot ensure that the process stabilizes on a unique chain. Nevertheless,  this difficulty can be avoided, identifying $(A, B)$  and $(0, 0)$, and, consequently,  working on  cycles instead of chains. 

all these results are formally written in  Section \ref{results}. 
We also conjecture that another way of getting a complete stabilization is to allow a larger sight  to sites close to endpoints.


\subsection{Contexts}

In this paper, all these inspirations are translated into a graphical interpretation in order to uniform the techniques issued from different fields and to ease the reading of the article. Thus, even if these inspirations are hidden in the current form of the article, they are nevertheless present and a reader familiar with these domains will be able to make the link. We now present the originality of our work compared to these different domains.

\subsubsection{Crystallography}
We were lead to consider this problem when studying a model of cooling process of crystals. Crystals are made of several kind of atoms. At high temperature, the structure of the atoms is chaotic but when the temperature decreases the atoms self-stabilize into a crystal, an ordered structure. Crystals are commonly modeled by tilings \cite{Henley} \textcolor{red}
In a set of studies \cite{FBR2010ap, FBR2010an,FBR2010ao}
 we developed a model which transforms an unordered tiling into an ordered one to model the cooling of a crystal. 

Here, the model represents the cooling of a crystal with two kinds of atoms disposed on a line: one kind is represented by an horizontal segment and the other one by a vertical segment. At high temperature, the segments are unordered and thus the corresponding thread is twisted. At low temperature, an atom wants the less possible neighboring atoms to be as the same kind as itself. The corresponding configuration in our model is the discretized line linking the two endpoints. We propose here an explanation on how atoms self-stabilize from a chaotic structure to an ordered one. The fact that the agents of this paper simulate atoms justify all the previous constraints. All of our anterior studies assume that the quantity of each kind of atom is the same. In particular, \cite{FBR2010an} deals with our problem when the line to approximate has a slope $1$ (the number of horizontal segments is equal to the number of vertical segments). The present  paper is a generalization of this previous study. This generalization is not direct since we show that the atoms need to consider a wider neighborhood to self-stabilize. 

\subsubsection{Distributed computing}

The robot chain problem \cite{DKLM06}  in distributed algorithms is really close to our problem except that the topology is continuous. Several solutions where presented to solve this problem but most of them, such as Manhattan-Hopper~\cite{KM09}, rely on unavailable information in our setting (robots have names, know global informations or can fuse together, \ldots) and thus are not applicable. Nevertheless, one algorithm is interesting for our case: the \textsc{Go-To-The-Middle} algorithm. In this algorithm, when a robot decides to act, it moves to the middle of the line linking its two nearest neighbors. In this paper, the movements of our "robot" are limited to jumping from one site to another one and thus for this community, our work can be seen as a discretization of \textsc{Go-To-The-Middle}. 

Note that this discretization of a continuous algorithm is far from trivial. The difficulty  comes from the fact that the discrete topology creates some artifacts. Dealing with these artifacts creates oscillations in the obtained dynamics, which does not appear in the continuous case. 

As a proof of this complexity, we show that if the space is discrete then the robots need to know the position of more than their two nearest neighbors to achieve our goal (see theorem \ref{main:impossible}). 

\subsubsection{Language theory}
Obtaining a good discrete approximation of a continuous line is a classical problem. It is important to note that this problem is heavily linked to language theory. Indeed a discrete line on the plane can be represented as a word where one letter represents an horizontal segment and the other one represents a vertical segment. Words representing a discrete approximation of a continuous line of rational slope are called Christoffel word \cite{Ber07} (Sturm words are a discrete approximation of continuous lines of irrational slope). In this paper, because of the graphical representation of the problem, we will not  use directly use the vocabulary  of  language theory, even if we strongly rely on it to develop our algorithm and to deal with the discrete artifacts. In particular, the simulator used to do Figure \ref{trace1} computes Christoffel sequences to determine the movements of  agents. 



\subsection{Tools}

\subsubsection{Height Functions on Tilings }
Lot of work was done to categorize all the different kinds of tilings \cite{Grunshe} . 
 A main tool in the study of   is the notion of height functions, introduced by W. Thurston \cite{Thu90} and independently in the statistical physics literature (see  \cite{Bur97} for a review). The height function is used to control  the evolution of a tiling of a given shape by a succession of local modifications (called flips). Nevertheless, any algorithm based on a height function is a centralized one, since the height function cannot be computed locally. 
 Here, we present a distributed version of this method. Our algorithm does not uses the height function, but the height function has a crucial role in the analysis.

\subsubsection{Probabilistic Cellular automata }

The   oscillations of the height function are related to the evolution  of a stochastic cellular automaton called \textsc{ECA} $178$ previously studied in \cite{FMST06}.  We have been able to use our expertise in the analysis  of the convergence  of probabilistic cellular automata to transfer it to  the analysis  of the convergence of our algorithm. By this way we can  ensure a polynomial time of coalescence in average. \\

In section \ref{Not}, we define formally our problem. In section \ref{def:rule}, we 
present our process and its  main properties. We  
show that the obtained dynamics quickly converges to a solution which is close to the objective line , and give our results in section \ref{analyse}.  
 Finally in section \ref{conc}, we present some  questions left open by this paper.

%
%
%

\section{Notations and definitions}
\label{Not}

\subsection{The model}

\subsubsection{Configurations and associated words}

Let $(A, B)$ be a pair of positive integers. We state ${\gcd(A,B)} = n$, $t_a = \frac{A}{n}$, $t_b = \frac{B}{n}$, $\per = t_a+t_b$ and  $\tot = A+B =  n\,  \per$.    
The ratio $ \frac{B}{A} =\frac{t_b}{t_a}$ represents the slope of the continuous line of equation: $-t_b x + t_a y = 0$,    passing by $(0, 0)$ and $(A, B)$, that  we wish to approximate. This line is called the \emph{ideal line.} 
In this paper, elements of $\N^2$ are called \emph{sites}.

We state: $a = (1, 0)$,  $b= (0, 1)$  and $\Sigma=\{a,b\}$. A \emph{configuration} $c$ is a sequence $(c_0, ...., c_\tot)$ of sites such that  $c_0 = (0, 0)$, $c_\tot = (A, B)$ and  for each $i <m,$ either 
$c_{i+1} - c_i = a$ or $c_{i+1} - c_i = b$. The set of configurations is denoted by $\Spath$. 
The word $\word$ associated to the configuration $\conf$ 
is  the word $\word = w_1 w_2 ... w_\tot$ of $\Sigma^{\tot}$,  such that; for each $i$ of $\{1, 2,...\tot\}$, $\word_i$  is the value of  $\conf_{i}  - \conf_{i-1}$. 

For each word $\word$ of $\Sigma^{*}$, let   $|\word|_a$ (respectively $|\word|_b$) denote  the number of letters $a$  (respectively $b$) in $\word$, and $|\word| = |\word|_a + |\word|_b$. 
If $\word$ is associated to a configuration $\conf$,  then we  have $|\word|_a=A$; $|\word|_b=B$. Conversely, 
for each word $\word$ of $\Sigma^{\tot}$,  with $ |\word|_a = A$  and $ |\word|_b =A$,  can be associated to  a unique configuration $\conf$ of $\Spath$, \emph{i.e.} $\Spath$ and $\{w \in \Sigma^{\tot}: |\word|_a=A \text{ and } |\word|_b=B \}$ are in bijection.

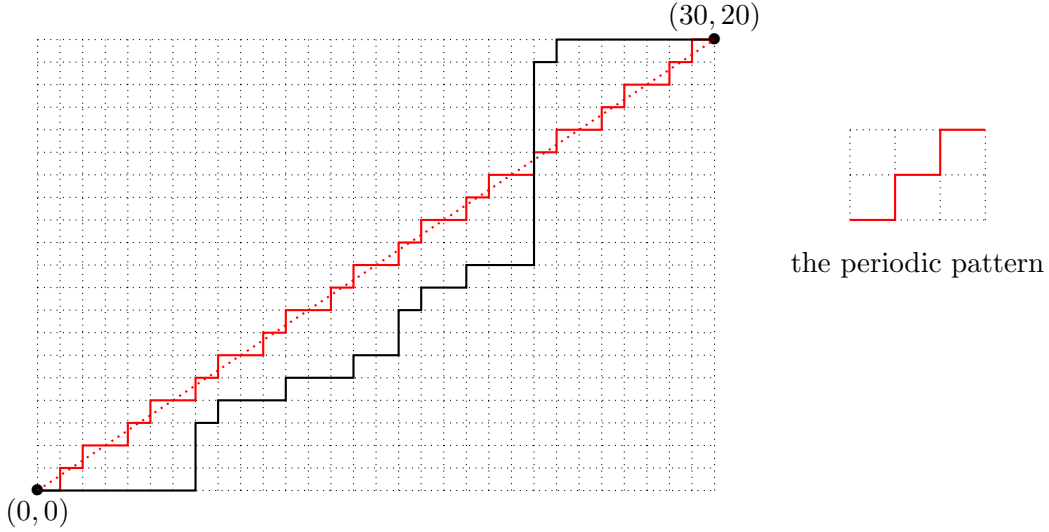
\begin{figure}[htb]
\begin{center}
\begin{tikzpicture}[x=0.3cm,y=0.3cm]

\path [dotted, draw, thin] (0,0) grid[step=0.3cm] (30,20);
\foreach \Point in {(0,0),(30,20)}{
    \node at \Point {\textbullet};
}
\node [black] at (0,-1) {$(0,0)$};
\node [black] at (30,21) {$(30,20)$};

\path [draw, color=red, thick] (0,0) -| (1,1) -| (2,2) -| (4,3) -| (5,4) -| (7,5) -| (8,6) -| (10,7)-| (11,8) -| (13,9)-| (14,10)-| (16,11)-| (17,12)-| (19,13)-| (20,14) -| (22,15)-| (23,16)-| (25,17)-| (26,18)-| (28,19) -| (29,20) -- (30,20);

\path [draw, red, dotted, thick] (0,0) -- (30,20);

\path [draw, thick] (0,0) -| (7,3) -| (8,4) -| (11,5) -| (14,6) -| (16,8) -| (17,9) -| (19,10)-| (22,19) -| (23,20)-- (30,20);


\path [dotted, draw, thin] (36,12) grid[step=0.6cm] (42,16);

\path [draw, color=red, thick] (36,12) -| (38,14) -| (40,16) -- (42,16);
\node [black] at (39,10) {the periodic pattern};

\end{tikzpicture}
\caption{The grid is drawn in dashed black, the continuous line linking the two endpoints is in dotted red, a good discrete approximation of this line is in red and a chain is in black. The parameters are $t_a=3$, $t_b=2$ and $n=10$. Thus $\per=5$, $\tot=50$ and the slope of the continuous line is $\frac{2}{3}$. This line is approximated by repeating a pattern of length $\per=5$ represented in the right part of the figure.}
\label{fig:grid}
\end{center}
\end{figure}

\subsubsection{Height and thickness }
For each site  $c =  (x, y)$, we define the \emph{height} $h(c)= -t_b x + t_a y$.  In the  following, we will extensively use the following properties: 

\begin{property}
Let $c$ and  $c'$ denote two sites. 
\begin{itemize} 
\item  $h(c + a) = h(c)  -t_b$ and $h(c + b) = h(c)  + t_a$,
\item we have $h(c) = h(c')$ if and only if there exists an integer  $k$ such that $c'-c = k(t_a, t_b)$,
\item we have   $h(c) \equiv h(c')\mod (\per)$ if and only if there exists two  integers  $k$  and $k'$ such that $c'-c = k(t_a, t_b) + k' (-a+b)$.

 In particular, for each $i$ of $ \{0,\ldots, \tot\}$, the value $h(c_i)\mod (\per)$ does not depend on the configuration $c$ but only on $i \mod (\per)$.
\end{itemize}

\end{property}
  
For a configuration $c$, we define $\lMin(c)   = \min\{h(c_i), 0 \leq i \leq \tot \}$ and $\lMax(c) = \max\{h(c_i), 0 \leq i \leq \tot \}$.  Thus the configuration $c$ is included  the closed strip limited by lines (of slope $\frac{t_b}{t_a}$) of equations $-t_b x + t_a y = \lMin(c)$  and $-t_b x + t_a y = \lMax(c)$. Moreover this strip is the smallest one among all strips limited by lines of slope $\frac{t_b}{t_a}$. 
The \emph{thickness} $\DDl(c)$ of the configuration $\conf$ is defined by  $ \DDl(c) =   \lMax(c) - \lMin(c)$. The properties of values  $\mod (t_a+t_b)$ imply that $ \DDl(c) \geq t_a+t_b-1$.

Among all configurations of $\Spath$, we will focus on the ones which are  good approximations of the finite continuous line linking $(0,0)$ to $(A, B)$, \emph{i.e.} configurations $c$ such that $\DDl(c)$ is minimal. The words associated to these configurations are called \emph{Christoffel words} and are extensively studied \cite{Ber07}. We use the following definition of Christoffel words, which is the most practical in our context: the word $w$ associated to a configuration $c$ is a Christoffel word of slope $\frac{t_b}{t_a}$ if and only if $\DDl(c) = t_a+t_b-1$.  We will also say that, in such a case, $c$ is a \textit{Christoffel configuration}. 

\subsection{Local transition rules }
Configurations are static objects, now we introduce a way to modify them locally. Consider a configuration $\conf$ and $2\leq i \leq \tot$, the configuration $\conf'$ obtained by \emph{flipping} letters $i-1$ and $i$ in $\conf$ is defined as follows: consider the words $\word, \word'$ of $\Sigma^\tot$ where $\word$ is the word associated to configuration~$\conf$ and $\word'$ is obtained by flipping letters $w_{i-1}$ and $w_i$ in $w$, \emph{i.e.} $\word'_{i-1}=\word_{i}$, $\word'_{i}=\word_{i-1}$ and for all $j \in \{1,\ldots, \tot\} \setminus \{i-1,i\}$, $\word'_{j}=\word_{j}$; then $\conf'$ is the configuration associated to $\word'$. Note that for all $0\leq j \leq \tot$ such that $j\neq i$ we have $\conf_j=\conf'_j$ and:
\begin{itemize}
\item if $\word_{i-1} = a$ and $\word_{i}=b$ then $\conf'_j=\conf_j-a+b$ and the flip is called \emph{increasing},
\item if $\word_{i-1} = b$ and $\word_{i}=a$ then $\conf'_j=\conf_j+a-b$ and the flip is called \emph{decreasing}.
\end{itemize}
We will also denote this operation as ``flipping in site $c_{i}$'' for a configuration $c$. Consider a configuration $c$, doing an increasing (resp. decreasing) flip in  $c_i$ increases (resp. decreases) the  height  of this site from $\per$ units. \\


We fix   a positive integer $\sight$. 
Let $\conf$ be a configuration, $\word$ its associated word, and $i$ such that $1 \leq i \leq \tot -1$.  The \emph{right word }  $w^r(c, i)$ in $i$ for $\conf$ is the word $w(c, i)$ of $\Sigma^{\leq s}$ defined by $w^r(c, i) =  w_{i+1}\, w_{i+2} \, ... w_{i+\sight}$ when $i+\sight \leq \tot$, and  $w^r(c, i) = w_{i+1}\, w_{i+2} \, ... w_\tot $ when $i+\sight > \tot$.
 In a similar way, the \emph{left word }  $w^l(c, i)$ in $i$ for $\conf$ is the word $w^l(c, i)$ of $\Sigma^{\leq s}$ defined by $w^l(c, i) = w_{i}\,w_{i-1} ... w_{i-\sight+1}$ when $i \geq \sight$, and   $w^l(c, i) = w_{i}\,w_{i-1} ... w_{1}$ when $i < \sight$. We choose to write  $ w^l(c, i) $ reversing indices, since we adopt the point of view of a processor located  in $c_i$, which reads words  starting from its own position. 
 
Let $\Sigma^{\leq s}$  denote the set of non-empty words  on $\Sigma$ of length at most $\sight$, \emph{i.e.}  $\Sigma^{\leq s} =  \bigcup_{s' = 1}^s\Sigma{s'}$. A \emph{local transition rule} of sight $\sight$ is given by a function $\delta : (\Sigma^{\leq s} )^2 \rightarrow \{0, 1\}$.  
 Given a configuration $c$, we say that the site $c_i$ is \emph{active}  in $\conf$ when $\delta( w^l(c, i), w^r(c, i)) =1 $. Otherwise, the site $c_i$ is \emph{inactive}  in $\conf$. 

Notice that from our formalism, the activity status of a  site $c_i$ does not depend on 
\begin{itemize}

\item  any global parameters: $t_a,t_b, n, \tot$: the rule is \emph{local},

\item  the integer $i$:  the rule is \emph{anonymous}, 

\item  the   position  of $c_i$ on the grid: the site is \emph{partially lost}.
\end{itemize}

We say ``partially lost'' since there exist rules which allow $c_i$ to use some local   elements of orientation:  the site does not know  its own position, but, nevertheless, it  can possibly make difference between the top and the bottom, and between  clockwise and counterclockwise senses, and use these informations to choose its activity status.  

 We want to use a rule where sites are ``completely lost'', in the sense that the rule do not use the informations above. Formally,  a rule  $\delta$ is  \emph{totally  symmetric}  when 

\begin{itemize}

\item for each pair $(w, w')$ of  $(\Sigma^{\leq s})^2$, we have $\delta( w, w')= \delta(w', w)$,

\item for each pair $(w, w')$ of  $(\Sigma^{\leq s})^2$, we have $\delta( w, w')= \delta( g(w), g(w')$, where $g$ is the word morphism on   $\Sigma^*$ such that $g(a) = b$ and $g(b) = a$
\end{itemize}

The first item claims the invariance of $\delta$ by the central symmetry, the second one claims the invariance of $\delta$ by the symmetry according to the main diagonal line.

%

By abuse of notation, for a configuration $\conf$ (which my be a random configuration),  we design by $\dyn(c)$ the random configuration obtained as follows: a number $i$ of $\{1, 2, \tot -1\}$ is selected uniformly at random, and if the corresponding site  $c_i$ is active,  then it is flipped,   otherwise nothing is done:  $\dyn(c) = c$.

The transition rule $\dyn$ introduces a discrete Markovian process  on configurations:  let $\conf^t$ design the configuration at time $t$, $\conf^0$ is the \emph{initial configuration}. The configuration at time $t+1$ is a random variable defined by~$\conf^{t+1}=\dyn(\conf^t)$. 

\section{The specific  transition rule}
\label{def:rule}
\subsection{Construction}
Our aim is to specify the rule $\delta$ in order to construct   a \emph{coalescence process},  \emph{i. e.} given by a totally symmetric local transition rule such that: 
\begin{itemize}
\item any initial configuration will reach a  Christoffel configuration of  slope $\frac{t_b}{t_a}$ in polynomial expected time,
\item all   Christoffel  configurations of slope $\frac{t_b}{t_a}$ are \emph{stable}, \emph{i. e.} have   no active site.
\end{itemize}

We will now describe our transition rule $\delta$. First, $\delta = \max\{ \delta^r, \delta^l \}$, where $ \delta^r$  and $ \delta^l$ are rules such that for each pair $(w, w')$ of words, we have $ \delta^r(w, w')  =  \delta^l(w', w)$. So it suffices to define $\delta_r$ to completely define $\delta$, and we are ensured that the first condition for $\delta$ to be totally symmetric really holds.

Notice that, when $w^r(c, i)_1= w^l(c, i)_1$, the fact of flipping in  $i$ is completely irrelevant.  Thus, we can  state:  $\delta^r (w,  w') = 0$ when $w_1 = w'_1$. 
In order to ensure     $\delta$ to be totally symmetric, we construct $\delta^r$ such that   $\delta^r (w,  w') = \delta^r (g(w),  g(w')) $. Thus,  it can be assumed without loss of generality that $w'_1 = b$.  
We have several  constraints to have $\delta^r (w,  w') = 1$.  If those there constraints are simultaneously satisfied then, $\delta^r (w,  w') = 1$. Otherwise $\delta^r (w,  w') = 0$. These constraints are stated and explained below. \\

\textbf{Sight constraint: }the first one is that:  
\begin{align} 
 |\word'|  = \sight.
\end{align} 
The interpretation is clear, $c_i$ must have a visibility at least $\sight$ on its right side. \\

\textbf{Weak thickness constraint:} this second constraint is the heart of the process. It ensures that the thickness of the configuration is not increasing wherever the flip is done. This is not trivial without the  knowledge of  $(t_a, t_b)$. 

For $ 1 \leq i \leq s$, we define $a'_i $ (respectively $b'_i$) as the number of $a$ (respectively $b$) in the prefix of length $i$ of $w'$, \emph{i.e.}  the word $ w'_1w'_2 ... w'_i$;  and for  $ 1 \leq j  \leq  \vert \word  \vert$, we define $a_j $ (respectively $b_j$) as the number of $a$ (respectively $b$) in the prefix of length $j$ of $w$. \emph{i.e.}  the word $ w_1w_2 ... w_j$. We define $(r_a, r_b)$ as $(a'_i, b'_i)$, with $ \frac{b'_i}{a'_i}$ minimum (with the convention $\frac{b}{0} =+ \infty)$. 
in case of tie,  we take the pair with the lowest index $i$ (according to the previous  convention, for   $w' = bb....b$,  we take $(r_a, r_b) = (1, 0)$). 

The second constraint to possibly have  $\delta^r (w,  w') = 1$ is:  
\begin{align} 
 \exists \, j \in \{ 1, 2, ..., \sight   \}  \, \vert  \, r_b a_j - r_a b_j \geq r_a+ r_b. 
 \label{thick}
\end{align} 

Notice that $-r_b x + r_a y \geq r_a+ r_b +h(c_i)$ is an equation of the half-plane limited by the line $\lR'$ of slope  $ \frac{r_b}{r_a}$  passing by $c_i-a+b$, and not containing $c_i$. The condition claims that there exists a site $c_{i-j}$, with $ 1 \leq j \leq s$, which is element of this half-plane, \textit{i. e.}  is over the limit line. 
\begin{definition}
Let $\sight$ be a positive integer.  
A pair  $(u, v) $ of $\N^2$ is \emph{visible} by $ \sight$ if $u+v \leq \sight$. 
\end{definition}

\begin{lemma}
\label{lem_haut}
Assume that $(t_a, t_b)$ is visible by $ \sight$. 
Let $c$ be a configuration and $i \in \{1, 2, ..., \tot -1\}$ such that, if we state  $(w^l(c,i), (w^r(c,i)) = (w, w') $, then  $(w, w')$ satisfies the two constraints above, and let $j$ be an integer allowing to satisfy the thickness constraint.  
\begin{itemize}
\item If $ \frac{r_b}{r_a} > \frac{t_b}{t_a}$,  then  $h(c_i) + t_a +t_b  \leq  h(c_{i+t_a +t_b})$,
\item If $ \frac{r_b}{r_a} < \frac{t_b}{t_a}$, then   $h(c_i) + t_a +t_b < h(c_{i-j})$,
\item If $ \frac{r_b}{r_a} =  \frac{t_b}{t_a}$, then $h(c_i) + t_a +t_b  \leq  h(c_{i-j})$. Moreover, in this case we have the equivalence: 
$$r_b a_j - r_a b_j  = r_a+ r_b \iff h(c_i) + t_a +t_b  = h(c_{i-j}). $$  
  \end{itemize}
\end{lemma}

\begin{proof}
If $ \frac{r_b}{r_a} > \frac{t_b}{t_a}$,  then,  first,  since $i \equiv (i+t_a +t_b) \mod(t_a+t_b)$, we have $h(c_i) \equiv  h(c_{i+t_a +t_b}) \mod(t_a+t_b)$. Thus,  it suffices to prove that $h(c_{i+t_a +t_b})>h(c_i)$.  
 We have $h(c_{i+\per}) = h(c_{i}) -  a'_t t_b + b'_t t_a$.  Thus, if  $a'_t = 0$ (which implies $ b'_t = \per$) then we are done. Otherwise,  we have $\frac{b'_t}{a'_t} \geq \frac{r_b}{r_a} > \frac{t_b}{t_a} $, which gives 
 $$h(c_{i+\per}) = h(c_{i}) -  a'_t t_b + b'_t t_a  =  
 h(c_{i}) +  a'_t(-  t_b + \frac{b'_t}{a'_t} t_a)  
 >  h(c_{i}) +  a'_t(-  t_b + \frac{t_b}{t_a} t_a) = h(c_i), $$
which gives the first item. \\

If $ \frac{r_b}{r_a} < \frac{t_b}{t_a}$, then, $ h(c_{i-j}) = h(c_{i} - a_j a -b_j b) = h(c_i) + t_b a_j  - t_a b_j .$  
 Thus,  we have to prove that  $t_b a_j  - t_a b_j > t_a+ t_b$, which can be rewritten in  $\frac{t_b}{t_a}(a_j-1) > b_j +1.$ On the other hand, the condition \ref{thick} can be be rewritten in  $\frac{r_b}{r_a}(a_j-1) \geq b_j +1$ (notice, that, with our convention,  $ \frac{r_b}{r_a} > \frac{t_b}{t_a}$ implies that $r_a \neq 0$). This ensures that $a_j-1 >0$. Thus, since $ \frac{r_b}{r_a} \leq \frac{t_b}{t_a}$, we obtain: 
 $$ \frac{t_b}{t_a}(a_j-1) >  \frac{r_b}{r_a}(a_j-1) \geq b_j +1. $$ 
which is the result. \\ 

If $ \frac{r_b}{r_a} =  \frac{t_b}{t_a}$, we proceed as in the second case to get 
$$ \frac{t_b}{t_a}(a_j-1) =  \frac{r_b}{r_a}(a_j-1) \geq b_j +1. $$ 
which gives the inequality. Moreover, we have $ \frac{r_b}{r_a}(a_j-1) =  b_j +1$ if and only if  $\frac{r_b}{r_a}(a_j-1) = b_j +1$, \textit{i. e.}    $ r_b a_j - r_a b_j  = r_a+ r_b$. 
On the other hand, $ \frac{t_b}{t_a}(a_j-1) = b_j +1$ if and only if  $t_b a_j  - t_a b_j = t_a+ t_b$,  \textit{i. e.}  $h(c_{i-j}) =  h(c_i) + t_a+ t_b$. This gives the equivalence. 
\end{proof}

\begin{center}
\begin{tikzpicture}[x=0.5cm,y=0.5cm]

\path [dotted, draw, thin] (0,0) grid[step=0.5cm] (12,8);

\path [dotted, draw, thin] (16,0) grid[step=0.5cm] (28,8);

\path [draw, thick] (1,1) -- (1,2) -| (4,3) -| (7,4) -| (9,5) -| (10,7) -- (11,7);

\path [draw, thick] (17,1) -- (17,2) -| (20,3) -| (23,4) -| (25,5) -| (26,7) -- (27,7);

\node [black] at (7.5,2.4) {$\textsc{$c_i$}$};
\node [black] at (5.5,4.4) {$\textsc{$c_i-a+b$}$};
\foreach \Point in {(7,3),(6,4)}{
    \node at \Point {$\bullet$};
}

\node [black] at (23.5,2.4) {$\textsc{$c_i$}$};
\node [black] at (21.5,4.4) {$\textsc{$c_i-a+b$}$};
\foreach \Point in {(23,3),(22,4)}{
    \node at \Point {$\bullet$};
}

\draw[color=red, thick] (0,0) -- (12,8);
\draw[thick] (0,1) -- (12,7);
\draw[thick] (1,0) -- (12,5.5);

\draw[color=red, thick] (16,2) -- (28,6);
\draw[color=red, thick] (16,0.66) -- (28,4.66);

\node [black] at (12.5,5.5) {$\lR$};
\node [black] at (12.5,7) {$\lR'$};
\node [black,color=red] at (-0.5,0) {$l'_{opt}$};

\node [black,color=red] at (15.5,2) {$l'_{opt}$};
\node [black,color=red] at (15.5,0.66) {$l_{opt}$};

\foreach \Point in {(1,2)}{
    \node  at \Point {$\bullet$};}

\node [black] at (0.5,2.5) {$\textsc{$c_{i-j}$}$};

\foreach \Point in {(25,5)}{
    \node at \Point {$\bullet$};
}
\foreach \Point in {(26,4)}{
    \node at \Point[color=red] {$\bullet$};
}

\node [black] at (25,5.5) {$\textsc{$c_{i+\per}$}$};
\node [black,color=red] at (26,3.2) {$(3,1)$};

\node [black] at (6,-1) {Case (a)};
\node [black] at (22,-1) {Case (b)};

\end{tikzpicture}

\end{center}

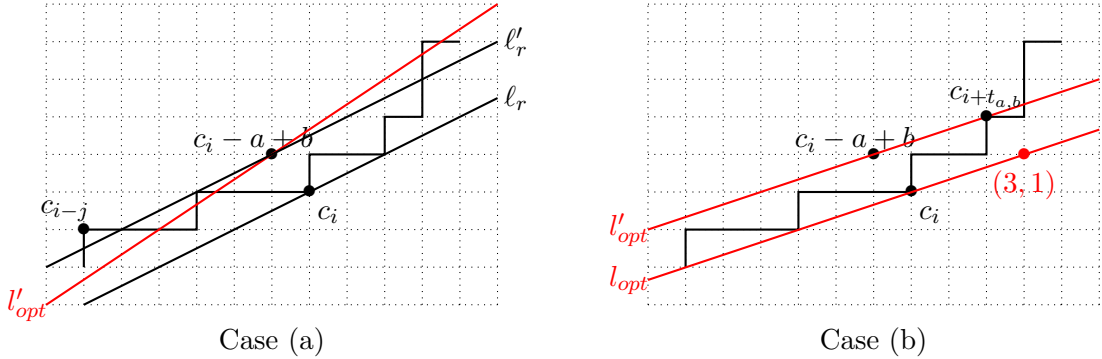
\captionof{figure}{The main ideas of  Lemma \ref{lem_haut}: the site $c_i$ has a sight of $8$, $(r_a, r_b) = (1, 2)$, thus   the slope of $\lR$ and $\lR'$is  $\frac{1}{2}$. In the case (a), the line $\lOpt'$ has a slope of $\frac{2}{3}$ and the  site $c_{i-j}$ is over $\lOpt'$. In the second case, the line $\lOpt'$ has a slope of $\frac{1}{3}$ and the  site $c_{i+ \per}$ is over $\lOpt'$. }

\begin{corollary}
\label{decrease}
For any configuration $c$, we have:  $ \lMin(c) \leq  \lMin(\delta(c)) \leq \lMax(\delta(c)) \leq \lMax(c)$,  and, therefore,   $\DDl( \delta(c)) \leq  \DDl(c)$. 
\end{corollary}

We also have the corollary below, noticing that in a site was flipped in a   Christoffel configuration, then either 
$\lMax$ would be  increased, which is not possible,   or $\lMin$ would be decreased, which is also impossible, by symmetry of the process.

\begin{corollary}
 Christoffel configurations are stable for any process satisfying the constraints above. 
\end{corollary}

\textbf{Strong thickness  constraint:} Assume now that previous constrains are both satisfied, and that, when a flip is done on $i$, then the new site indexed by $i$  is of maximal height. This may happen when $h(c_i ) + \per = \lMax$. 
If the maximal height appears in $i$, then the closest indices where the maximal height can eventually also be reached are $i+ \per$ and  $i- \per$, because of congruence conditions of Property 1. 
We want to be sure that our process creates no isolated maximum: if  the  maximal  height  $\lMax$ is reached in $i$, then $i$ is not an isolated maximum, in the sense of   $h(c_{i+ \per}) = \lMax $ or $h(c_{i- \per})  = \lMax$.  This is useful in the analysis, for energy compensations. 

But in the same time, we want to allow a sufficient instability to the process in order to make it move to a better configuration.  This is ensured by enforcing the weak thickness constraint as follows.
\begin{align} 
\exists   j \in \{ 1, 2, ..., \sight   \}  \,\, \vert  \,\, a_{j} - r_a b_{j} > r_a+ r_b \lor  (a_{j} - r_a b_{j} = r_a+ r_b \land gcd(a_{j }-1, b_{j}+1) = 1).
 \label{gcd}
\end{align} 
The strong constraint  adds that if all sites  $c_{i-j}$ are in the limit line, directed by $(r_a, r_b)$ passing through the site $c_i-a+b$, then there exists such a site such that  the components of the vector $c_i-a+b - c_{i-j}$ are relatively prime. 

If this strong constraint is satisfied, then the  weak constraint is automatically satisfied.  Nevertheless we prefer to present the process by this way, in order to have a real understanding of the motivations of the rules.   

\begin{lemma}\label{isole}
Assume that $(t_a, t_b)$ is visible by $\sight$. 
Let $c$ be a configuration and $i \in \{1, 2, ..., \tot -1\}$ We state  $(w^l(c,i), (w^r(c,i)) = (w, w') $. 
  Assume that $(w, w')$ satisfies the three constraints above and $h(c_i ) + \per = \lMax(c)$.
 
Then $h(c_{i+ \per}) = \lMax(c) $ or $h(c_{i- \per})  = \lMax(c)$. 
\end{lemma}

\begin{proof}
Lemma \ref{lem_haut} directly gives the result when  $ \frac{r_b}{r_a}  >  \frac{t_b}{t_a}$, and, from Lemma \ref{lem_haut} the hypotheses cannot occur when $ \frac{r_b}{r_a}  <  \frac{t_b}{t_a}$.
 Thus it remains to study the case when  $ \frac{r_b}{r_a} =  \frac{t_b}{t_a}$ and  $h(c_i) + \per = h(c_{i-j})$, which ensures that  $r_b a_{j} - r_a b_{j} = r_a+ r_b$, from Lemma \ref{lem_haut}. 
 
  In this case  $h(c_{i-j} +a-b)= h(c_{i-j} ) -t_b -t_a= h(c_i)$, thus, from Property 1, there exists an integer $k$ such that  $c_{i-j} +a-b- c_i = k(t_a, t_b)$, \emph{i.e.}  $(-a_{j} +1, -b_{j }-1) = k(t_a, t_b)$.  We have  $-b_{j} -1 < 0$ and,  from the relative primarity  constraint, $gcd(a_{j }-1, b_{j }+1) = 1$. Thus,  we necessarily have $k = -1$, which gives that
$j   =  a_{j }+ b_{j } =  a_{j }-1 + b_{j }+1 = t_a +t_b  = \per$.  Thus  $h(c_i) +\per = h(c_{i-\per})$. 
\end{proof}

%
%

\section{Analysis}
\label{analyse}

\subsection{The energy lemma}

We start by presenting the lemma used to prove time efficiency of our process. Lemma~\ref{lem:arg:mart} is a classical result on martingales, its proof can be found in~\cite{FMST06}. The way to use this lemma is to affect a value between $0$ and $k$ (with $k \in \mathbb{N}$) to each configuration, this value will be called the \emph{energy} $E(c)$ of configuration $c$. If wisely defined, this energy will behave as a random walk: its expected variation will be less than $0$ for any configuration. A non-biased one dimensional random walk on $\{0,\ldots, k\}$ hits the value $0$ on $O(k^2)$ time step. Once again if the energy is wisely defined, when the energy function hits $0$ then an irreversible update towards a stable configuration is done and by repeating this reasoning, we show that our process hits a stable configuration in polynomial time. A key part of this lemma is to bound the expected variation of energy, so we introduce the following notation:
$$\DD{E}(\conf^t) = E(\conf^{t+1})-E(\conf^{t}) \text{ (or } \DD{E}(\conf) = E(\dyn(\conf))-E(\conf)).$$

\begin{lemma} \label{lem:arg:mart}
Let $k \in \mathbb N$ and $\epsilon >0$. Consider $(\conf^t)_{t \geq 0}$ a random sequence of configurations, and $E:\Spath\rightarrow \N $ an \emph{energy} function. Let  $T=\min\{t: E(\conf^t) = 0\}$  be the the random variable which denotes the first time $t$ where $E(\conf^t) = 0$.
Assume that,   for  any $c$ such that $E(\conf) > 0$, we conjointly have:
\begin{itemize}
\item $E(c) \leq k$, 
\item $\expect[\DD{E}(\conf)|\conf] \leq 0$,
\item $Prob\{|\DD E(\conf)|\geq 1\}\geq \epsilon$.  
\end{itemize}
Then, $\expect[T] \leq \frac{k E (c^0)}{\epsilon}$   
\end{lemma}

\subsection{Our specific energy}
Our strategy consists in using Lemma \ref{lem:arg:mart} for an ``ad hoc'' energy function, that we will define now. 
Fix a configuration $\conf^0$ such that $\lMax(\conf^0) \geq t$.  The energy ${E}(c)$ of any configuration $c$ is defined as follows. 
If  $\lMax(c) \neq \lMax (\conf^0)$, 
 then ${E}(c) = 0$.  If  
 $\lMax(c) =  \lMax(\conf^0)$,  
 then consider the set 
$$Border^+ = \{i \in  \{1 , 2, ... \tot -1 \}, \exists i_0    \in  \{1 , 2, ... \tot -1 \} \vert h(\conf^0_{i_0}) =  \lMax(\conf^0) \mbox{ and  } i \equiv  i_0 \mod (\per) \}. $$

We recall that if $i$ and $j$ are both elements $Border^+$, then  $i \equiv j \mod (\per) $. Remark that $n-1 \leq \vert Border^+ \vert  \leq n$. We define the following sets : 

$$Top^+(c)   =   \{ i    \in  \{1 , 2, ... \tot -1 \} \vert h(\conf_{i}) =  \lMax(\conf^0) \} $$   
$$Down^+(c)  =  \{  i \in  Top^+(c), i+\per  \leq \tot,  i+\per \notin Top^+(c) \}$$
 $$Up^+(c) =   \{ ( i \in  Top^+(c), i-\per \geq 0 \in   i-\per \notin Top^+(c) \}$$

Notice that we have $Top^+(c) \subseteq Border^+$. 
The energy ${E}(c)$ of the configuration $c$ is the sum: 
$${E}(c) =  2 \vert Top^+(c)  \vert + \vert Down^+(c)  \vert+ \vert Up^+(c)  \vert$$


\begin{prop} 
\label{apply}
When  $\lMax(\conf^0) \geq t$, 
the energy defined above satisfies hypotheses of Lemma \ref{lem:arg:mart} with $\epsilon = \frac{1}{\tot-1}$, and $k = 3n$ 
 \end{prop}

 We first need the following lemma.

\begin{lemma}
\label{active}
 Let $c$ be a configuration. Assume  that there exists $i \leq \tot -\per$  and $j$ 
such that $ h(c_{i}) =  \lMax(c)  $ (respectively  $ h(c_{i}) =  \lMin(c)  $), 
 and there exists $j$ such that $ 0 < j \leq \per$ and  $ h(c_{i-j})  + \per \leq  \lMax(c)$   (respectively    $ h(c_{i-j})  - \per \geq  \lMin(c)  $).
 
 Then, the site $c_i$ is active in $c$. 
\end{lemma}
 
 \begin{proof} By symmetry, it suffices to prove it for the minimum case. 
If $ h(c_{i}) =  \lMin(c) $, the first letter or  $w^r(c, i)$ is $b$, and the first letter or  $w^l(c, i)$ is $a$. 
  So, we state $ w = w^l(c, i)$, $w' =g( w^l(c, i)$,  and we have to prove that $\delta^r(w, w') = 1$, \textit{i. e.} the pair
  $(w, w')$ satisfies the  constraints. 
  
  First, the  hypothesis $i \leq \tot - \per$ ensures that the sight constraint is satisfied. 
 Afterwards, 
 we have  
 $ h(c_{i-j}) = h(c_{i} - a_j a -b_j b) = h(c_i) + t_b a_j  - t_a b_j $. 
 
If there exists $j < s$ satisfying the hypothesis, then, from  Property 1,   $ h(c_{i-j})  - \per \geq  h(c_{i})$ gives actually the strict inequality:  $ h(c_{i-j})  - \per >  h(c_{i})$.  
  Thus,  $ h(c_{i})  +  t_b a_j  - t_a b_j  -\per >   h(c_{i})$, which gives   $t_b a_j  - t_a b_j > t_a+ t_b$, which can be rewritten in $\frac{t_b}{ t_a}(a_j -1) > b_j +1$. 
  On the other hand, since $ h(c_{i}) =  \lMin(c) $, we have $ \frac{r_b}{ r_a} \geq \frac{t_b}{ t_a}$. 
 Therefore, we get: $\frac{r_b}{ r_a}(a_j -1) > b_j + 1$, \textit{i. e.}  $r_b a_j  - r_a b_j > r_a +r_b$: the strong thickness constraint is satisfied.
 
 If $s$ is the only possible $j$ satisfying the hypothesis, then we have two alternatives. Either  $ h(c_{i-s})  - \per >  h(c_{i})$, and the  arguments of the paragraph just above  can be used to conclude, or   $ h(c_{i-s})  - \per  =   h(c_{i})$. The latter alternative  gives  $t_b a_s  - t_a b_s = t_a+ t_b$, which can be rewritten in $\frac{t_b}{ t_a}(a_s -1) =  b_j +1$. Since $ h(c_{i}) =  \lMin(c) $, we have $ \frac{r_b}{ r_a} \geq \frac{t_b}{ t_a}$. If $ \frac{r_b}{ r_a} > \frac{t_b}{ t_a}$, then we get $\frac{r_b}{ r_a}(a_s -1) >  b_j +1$, \textit{i. e.}  $r_b a_j  - r_a b_j > r_a +r_b$: the strong thickness constraint is satisfied.  
 
 If $ \frac{r_b}{ r_a} = \frac{t_b}{ t_a}$, then we get $\frac{r_b}{ r_a}(a_s -1) =  b_j +1$, \textit{i. e.}  $r_b a_s  - r_a b_s = r_a +r_b$.  On the other hand
 $t_b a_s  - t_a b_s = t_a+ t_b$ give  $t_b (a_s -1) = t_a (b_s +1)b$, thus, since $gcd(t_a, t_b) = 1$, there exists a positive integer $k$ such that 
 $a_s -1 = k t_a$  and $b_s +1 = k t_b$. This gives $a_s -1 +  b_s +1= k t_a +  k t_b$. But we know that   $a_s +  b_s =  t_a +   t_b$, thus $k =1$, $a_s -1 = t_a$  and $b_s +1 = t_b$.  Thus $gcd(a_s -1, b_s +1) = gcd(t_a, t_b) = 1$: the strong thickness constraint is satisfied. 
 \end{proof}

%
%
%

 We can now prove Proposition \ref{apply}. 
 
\begin{proof}

One easily sees that ${E}(c) <  2n$ since $\vert Top^+(c)  \vert  \leq n$ and $\vert Down^+(c)  \vert+ \vert Up^+(c)   \vert  \leq n-1$ and at least one equality must be strict. This gives the first item. \\

The second item is trivial, since, by definition,   $ Top^+$ is not empty. The site  $c_i$ of  $Top^+$  of lowest index is active, from lemma \ref{active}. (notice that we need the hypothesis : $\lMax(\conf^0) \geq t$  to ensure it,  in the case when $i < t$), and when $i$  is randomly chosen, with probability $\frac{1}{\tot-1}$,  the energy decreases from at least 1 unit  (actually 2 units,  except  when $i < \per$ and $i+\per \in Top^+(c))$, or  $i > \tot - \per $ and $i-\per  \in Top^+(c))$. \\

For the third item, we need more notations. 
Let $c$ be a configuration with positive energy.  For each $i \in \{1, 2, ..., \tot -1\}$,
denote by $\delta_i(c)$ be the configuration deduced  from $c$, when $i$ is chosen by the random process. 


 
Now,  make a partition $P^+(c)$ of  $Border^+(c)$ in subsets of at most three consecutive elements in such a way that,
\begin{itemize}
\item  for each $i \in  Down^+(c)$, integers $i$ and  $i+\per$  are in the same subset,
\item for each $i \in Up^+(c)$ such that  $i -2t \notin Top^+(c)$, then integers $i$ and  $i-\per$  are in the same subset. 
\end{itemize}

We claim that, for each subset $S$ the contribution of elements of $S$ to the value of  $\expect[\DD{E}(\conf)|\conf]$ is not positive. Indeed, let $S$ be an element of the partition $P^+(c)$.

\begin{itemize}
\item if $S = \{i\}$ and $h(c_i) < \lMax(c)$, then, by definition of  $P^+(c)$, we have  $h(c_{i-t} ) < \lMax(c)$ and $h(c_{i+t}) < \lMax(c). $ Thus,  from  lemma \ref{isole}, $h(\delta_i(c)_i) <   \lMax(c)$ thus $E(\delta_i(c)) = E(c)$. 

 If $S = \{i\}$ and $h(c_i) =  \lMax(c)$, we obviously have $E(\delta_i(c)) \leq  E(c)$, whether $i$ is active or not. Thus 
$$E(\delta_i(c)) \leq E(c). $$

\item if $S = \{i, i+\per\}$,   then assume without loss of generality that $h(c_i) = \lMax(c)$ and $h(c_{i+t}) < \lMax(c)$ (the other case is symmetric). 
Then $c_i$ is active,  $h(\delta_i(c)_i) <   \lMax(c)$, thus $E(\delta_i(c)) \leq  E(c) -2$. 
On the hand,   either $c_{i+t}$ is inactive and $E(\delta_{i+\per}(c)) =  E(c), $ or $c_{i+\per}$ is active and $E(\delta_{i+\per}(c)) \leq  E(c) +2$. Thus 
$$E(\delta_{i}(c)) +  E(\delta_{i+\per}(c)) \leq 2 E(c)$$

\item if $S = \{ i-t, i, i+t\}$,   then we have   $h(c_i) = \lMax(c)$,  $h(c_{i- t}) < \lMax(c)$ and  $h(c_{i+t}) < \lMax(c)$. 
Then $c_i$ is active,  $\delta_i(c)_i <   \lMax(c)$, thus $E(\delta_i(c)) =  E(c) -4$.

On the other hand,  as in the previous case, we have  $E(\delta_{i+t}(c)) \leq  E(c) +2$. By symmetry, we also have $E(\delta_{i- t}(c)) \leq  E(c) +2$. Thus we get 
$$E(\delta_{i- t}(c)) +  E(\delta_{i}(c)) + E(\delta_{i+t}(c)) \leq 3 E(c)$$
\end{itemize}

We have, from   lemma \ref{lem_haut}
 $$\expect[E(\delta(c)] = \frac{1}{\tot -1} \sum_{i \in \{1, 2, ..., \tot -1\}}E(\delta_{i}(c)  =   \frac{1}{\tot -1} \sum_{i \in Border^+} E(\delta_{i}(c))  $$
 
 Using our partition, we get
$$\expect[E(\delta(c)]   = \frac{1}{\tot -1} \sum_{S \in P^+} \sum_{j \in S}E(\delta_{i}(c)) 
  \leq   \frac{1}{\tot -1} \sum_{S \in P^+} \vert S \vert E(c)  = 
  E(c)  (\frac{ \vert Border^+ \vert}{\tot -1})  \leq E(c),$$
  
 which is the result.  
%
%
\end{proof}
%

\subsection{Results}
 \label{results}

\subsubsection{ Nonnegative configurations} 

We say that a configuration is nonnegative if for each $i \in \{0, 1, \tot \}$ we have $h(c_i) \geq 0$.

\begin{theorem}
If $(t_a, t_b)$ is visible by $\sight$, and the configuration $c^0$ is nonnegative, then the random process is a coalescence process in time  $0(n^4)$ time units in average. The  configuration reached is the unique Christoffel  configuration $c_{[0, \per-1]}$ such that $\lMax(c_{[0, \per-1]}) = \per -1$ and $\lMin(c_{[0, \per-1]}) = 0$. 
\end{theorem}

\begin{proof}
We can decompose  $T = \sum_{j = \per}^{2n-1} T_j$ where $T_j $ is the time to get a configuration $c$ such that 
 $\lMax(c) \geq j-1$ from a configuration $c'$ such that $\lMax(c) \geq j$. From Proposition \ref{apply}, we have 
 
 $\expect[T_j] \leq (2n-1)^2 (\tot -1) $, thus  $\expect[T] \leq (2n-1)^3 (\tot -1) $. Moreover,  from corollary 1 we have, for any $t$,  $\lMin(c^t) \geq 0$. Thus,  after time $T$, we get a configuration $c$ such that, $\lMax(c) \leq  \per -1$ and $\lMin(c)  \geq 0$ which ensures that $c =  c_{[0, \per-1]}$. We know that $c_{[0, \per-1]}$ is stable, from corollary 2. 
 \end{proof}


\subsubsection{ The general bounded case }

If we work with two energies, one as described above, related  to $\lMax$, and one symmetric, related to $\lMin$, one  gets, in a similar way: 

\begin{theorem}
\label{general}
If  $(t_a, t_b)$ is visible by $\sight$, then, with any initial configuration,  the random process almost surely reaches a configuration $c$  such that:  $ - \per +1 \leq \lMin(c)  < \lMax(c) \leq  \per - 1$.   The time $T$ necessary to reach such a configuration in $0(n^4)$ in average. 
\end{theorem}

Notice that our arguments fail to continue to decrease the thickness, because of difficulties at the boundary. For $\lMax(c) \leq  \per - 1$, when the lowest element $i$ of $ Top^+(c)$ is such that $i \leq \per$, we can cannot ensure that $i$ is active.

Thus Theorem \ref{general} is partially  satisfying: we reach a set of configurations which are only partially stable. Sites $c_i$ with $i \equiv 0 \mod (\per)$ are no more active and are  the ideal line of equation  $-t_b x+ t_a y = 0$.  But other some other site  can be. Nevertheless, one can remark that these active sites  only have the freedom to oscillate around the deal line, between the two postions which are the closest ones to the ideal line,\textit{ i. e. } the positions  of  lowest positive height, and of largest negative height.

It is not possible to always get the optimal thickness  with our algorithm.  For example, for $\sight = 5$ and $(t_a, t_b) = (3, 2)$, consider  the configuration $c$ associated  word $ (ba^2ba)^{n-1} ba^3b$. 
We have  $h(c_{\tot -1}) = \lMin(c)   = -3$ and $h(c_1) =  \lMax(c) = 3$. Moreover, with  Lemma \ref{decrease}, one  can easily see that,  for any integer $t >0$, $h( \delta^t(c)_1)  = 3$ and $h( \delta^t(c)_{\tot -1}) =  - 3$ (see Figure \ref{fig:example}). 

\begin{figure}[htbp] 
   \centering
   \includegraphics[width=10cm]{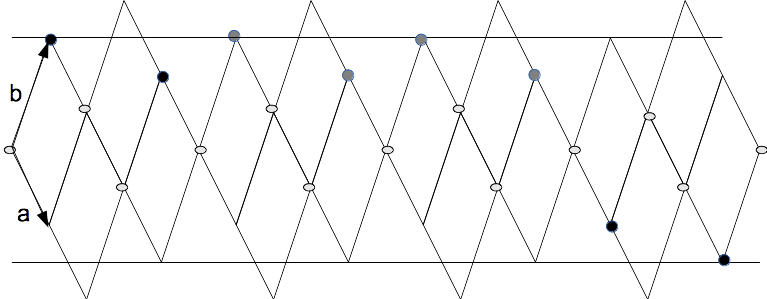} 
   \caption{An example of configuration whose thickness cannot be decreased. For any positive $t$, white sites cannot be active in $ \delta^t(c)$ because of the non increase of thickness. Green sites may be active in $ \delta^t(c)$,  but  no  choice of possible position for these sites allow black sites to be active.}
   \label{fig:example}
\end{figure}

\subsubsection{ The general periodic  case}

Notice that  configurations can be can be seen as cycles by identifying site $c_0$ and site $c_{\tot}$.  We call it the \emph{cyclic model}. 
Formally,  instead of considering sites as elements of $\Z^2$, they are considered as elements of the quotient space $\Z^2/ n(t_a, t_b) \Z$. This makes two main differences:  the site $c_0$ can be possibly active, and for each configuration $c$ and each index $i$, we have $\vert w^l(c, i) \vert = \vert w^r(c, i)  \vert = \sight$, so the sight constraint becomes irrelevant.

Using a very light modification of the energy function (adding two units for the energy when  $Top^+(c) = Border^+$), we obtain the following result.

\begin{theorem}
In the cyclic model, if  $(t_a, t_b)$ is visible by $\sight$, from any origin configuration $c^0$,  the random process is a coalescence process whose coalescence time $T$ in $0(n^4)$ in average.
\end{theorem}

 \subsubsection{ Impossibility result } 
 
\begin{theorem}
\label{main:impossible}
Consider any local rule $\delta$ of sight $\sight$, and take $(t_a, t_b) = (s+1, 1)$.  One of the following alternatives holds: 
\begin{itemize}
\item a Christoffel configuration of slope $\frac{1}{s+1}$ is not stable, 
\item  for any $k >0$, there exists a configuration $c$ such that $\DDl(c) \geq k $ and $c$ is stable. 

\end{itemize}

\end{theorem}

\begin{proof}
Consider 
the words $w=a^{s+1}b$, $w'=a^{s}b$, $w''=a^{s+2}b$, the configuration $c$ corresponding to $w^{2k}$ 
and $c'$ corresponding to $w(w'')^{k-1}(w')^{k-1}w$. 

The  configuration $c$ corresponds to a Christoffel word. Assume that  $c$ is stable under the local rule $\delta$. Thus,  $\delta(a^s,ba^{s-1}) = \delta(ba^{s-1},a^s) = 0$. This implies that  $c'$  is also stable under the local rule $\delta$.
On the other hand, we have  $\lMin(c')   = -s -k+1 $ and $\lMax(c) = 0$, thus  $\DDl(c) =  s+ k-1$. 
\end{proof}

\section{Conclusion and open questions}
\label{conc}

In this part, we start by presenting the improvements which can be done to this paper. Then we focus on the possible extensions and applications of our work.

We have introduced and analyzed a random process which enables a twisted thread to reorganize itself. We think that the core rule is optimal in term of sight and convergence speed in our setting. Nevertheless, we think that our analysis is not optimal. For the case of slope $1$, our random process is the same one as the one introduced and analyzed in $\cite{FBR2010an}$ but the analysis of this paper gives an upper bound on the convergence time of $O(\Dhz \tot^3)$ 
whereas in \cite{FBR2010an} they prove an upper bound of $O(\tot^3)$. We were able to generalize the random process for any rational slope but not the analysis. In fact, we conjecture that our random process converges in $O(\tot^3)$ since our analysis considers only the sites of maximal height and forgets about a lot of useful updates which are done in parallel. Also, we think that our rule for synchronizing the endpoints is not optimal in time and sight. We conjecture that both endpoints can be synchronized in polynomial time according to $n,s$ and $\per$ with agents of sight $2s$ at the endpoints. 

Another interesting question is to generalize our process to dimensions greater than two, \emph{i.e.} to an alphabet with more than two letters for the language theory version of this problem. In ongoing works, our process is working well experimentally in greater dimensions if given a big enough sight but we are not able yet to prove that there is no interlocking between the letters. This extension is interesting for two applications of our work.

The first application is for studying a model of cooling processes in crystallography \cite{FBR2010ap, FBR2010an, FBR2010ao}. In this paper, we study in fact a ``simple" case where two kinds of atoms are disposed on a line and these atoms want to diminish the interactions with the atoms of the same kind. Increasing the dimension of this problem corresponds to consider more than two kinds of atoms. Also, note that we previously analyzed the $2D$ case, with a periodic tiling in \cite{FBR2010ao} but our study supposes that the three kinds of atoms required in this tilling have the same proportions. It would be interesting to study this case, using our new method when atoms do not have the good proportion. Also, for concluding our set of studies, we would like to present and analyze a cooling model for a $2D$ aperiodic tiling like Penrose tiling. Aperiodic tilings correspond to quasicrystal and actually fabricating a quasicrystal of large size is not possible because the cooling process is not well understood.

The second application would be to generalize the density classification problem \cite{Pac88}. In this problem, we consider a one dimensional chain of agents. There are two states and each agents can memorize only one state. Using a distributed algorithm, agents must determine the majority state in the initial configuration while storing only one state by agent. It is known that this problem is not solvable under parallel dynamics \cite{LB95} but recently Fat\`es \cite{Fates11} solved this problem with any arbitrary precision using a probabilistic dynamics. We think that our result can be used to generalize the density classification problem with more states (which is again equivalent to increasing the number of dimensions) and to consider questions like "Is the initial density of state $a$  more than $\frac{B}{A}$?".

%
%
%
%
%

\bibliographystyle{plain}
\bibliography{MFCS2014}

%

\end{document}